\documentclass[runningheads,a4paper]{llncs}

\usepackage{amssymb}
\usepackage{graphicx}

\usepackage{url}
\urldef{\emailpsu}\path|{ furer, hwyu }@cse.psu.edu|

\newcommand{\keywords}[1]{\par\addvspace\baselineskip\noindent\keywordname\enspace\ignorespaces#1}
\usepackage{comment}
\usepackage[toc,page]{appendix}

\newtheorem{Proposition}{Proposition}

\usepackage{algorithm}
\usepackage{algorithmicx}
\usepackage{algpseudocode}

\usepackage[margin = 1.6in]{geometry}

\begin{document}

\mainmatter  


\title{Space Saving by Dynamic Algebraization}


%
%
\author{Martin F\"{u}rer
\and Huiwen Yu}
%


\institute{Department of Computer Science and Engineering\\
 The Pennsylvania State University, University Park, PA, USA.\\
\emailpsu }


%
%
\maketitle

\begin{abstract}
Dynamic programming is widely used for exact computations based on tree decompositions of graphs. However, the space complexity is usually exponential in the treewidth. We study the problem of designing efficient dynamic programming algorithm based on tree decompositions in polynomial space. We show how to construct a tree decomposition and extend the algebraic techniques of Lokshtanov and Nederlof \cite{savespace2010} such that the dynamic programming algorithm runs in time $O^*(2^h)$, where $h$ is the maximum number of vertices in the union of bags on the root to leaf paths on a given tree decomposition, which is a parameter closely related to the tree-depth of a graph \cite{treedepth}. We apply our algorithm to the problem of counting perfect matchings on grids and show that it outperforms other polynomial-space solutions. We also apply the algorithm to other set covering and partitioning problems.

\keywords{Dynamic programming, tree decomposition, space-efficient algorithm, exponential time algorithms, zeta transform}
\end{abstract}

\section{Introduction}

Exact solutions to NP-hard problems typically adopt a branch-and-bound, inclusion/exclusion or dynamic programming framework. While algorithms based on branch-and-bound or inclusion/exclusion techniques \cite{polyspace13} have shown to be both time and space efficient, one problem with dynamic programming is that for many NP-hard problems, it requires exponential space to store the computation table. As in practice programs usually run out of space before they run out of time \cite{openproblem}, an exponential-space algorithm is considered not scalable. Lokshtanov and Nederlof \cite{savespace2010} have recently shown that algebraic tools like the zeta transform and M\"{o}bius inversion \cite{mobiusorigin,stanley2000enumerative} can be used to obtain space efficient dynamic programming under some circumstances. The idea is sometimes referred to as the coefficient extraction technique which also appears in \cite{spaceicalp08,spaceicalp09}.

The principle of space saving is best illustrated with the better known Fourier transform. Assume we want to compute a sequence of polynomial additions and multiplications modulo $x^n-1$. We can either use a linear amount of storage and do many complicated convolution operations throughout, or we can start and end with the Fourier transforms and do the simpler component-wise operations in between. Because we can handle one component after another, during the main computation, very little space is needed. This principle works for the zeta transform and subset convolution \cite{fouriermobius} as well.

In this paper, we study the problem of designing polynomial-space dynamic programming algorithms based on tree decompositions. Lokshtanov et al. \cite{kpath} have also studied polynomial-space algorithms based on tree decomposition. They employ a divide and conquer approach. For a general introduction of tree decomposition, see the survey \cite{discovertw}. It is well-known that dynamic programming has wide applications and produces prominent results on efficient computations defined on path decomposition or tree decomposition in general \cite{dptw}. Tree decomposition is very useful on low degree graphs as they are known to have a relatively low pathwidth \cite{pathwidth}. For example, it is known that any degree 3 graph of $n$ vertices has a path decomposition of pathwidth $\frac{n}{6}$. As a consequence, the minimum dominating set problem can be solved in time $O^*(3^{n/6})$\footnote{$O^*$ notation hides the polynomial factors of the expression.}, which is the best running time in this case \cite{mindominatingset}. However, the algorithm trades large space usage for fast running time.

To tackle the high space complexity issue, we extend the method of \cite{savespace2010} in a novel way to problems based on tree decompositions. In contrast to \cite{savespace2010}, here we do not have a fixed ground set and cannot do the transformations only at the beginning and the end of the computation. The underlying set changes continuously, therefore a direct application on tree decomposition does not lead to an efficient algorithm.
We introduce the new concept of zeta transforms for dynamic sets. Guided by a tree decomposition, the underlying set (of vertices in a bag) gradually changes. We adapt the transform so that it always corresponds to the current set of vertices. Herewith, we might greatly expand the applicability of the space saving method by algebraization.

We broadly explore problems which fit into this framework. Especially, we analyze the problem of counting perfect matchings on grids which is an interesting problem in statistical physics \cite{monomer}.
There is no previous theoretical analysis on the performance of any algorithm for counting perfect matchings on grids of dimension at least 3. We analyze two other natural types of polynomial-space algorithms, the branching algorithm and the dynamic programming algorithm based on path decomposition of a subgraph \cite{pathwidthsparsegraph}.  We show that our algorithm outperforms these two approaches. Our method is particularly useful when the treewidth of the graph is large. For example, grids, $k$-nearest-neighbor graphs \cite{knn} and low degree graphs are important graphs in practice with large treewidth. In these cases, the standard dynamic programming on tree decompositions requires exponential space.

The paper is organized as follows. In Section 2, we summarize the basis of tree decomposition and related techniques in \cite{savespace2010}. In Section 3, we present the framework of our algorithm. In Section 4, we study the problem of counting perfect matchings on grids and extend our algorithmic framework to other problems.

\section{Preliminaries}

\subsection{Saving space using algebraic transformations}

Lokshtanov and Nederlof \cite{savespace2010} introduce algebraic techniques to solve three types of problems. The first technique is using discrete Fourier transforms (DFT) on problems of very large domains, e.g., for the subset sum problem. The second one is using M\"{o}bius and zeta transforms when recurrences used in dynamic programming can be formulated as subset convolutions, e.g., for the unweighted Steiner tree problem. The third one is to solve the minimization version of the second type of problems by combining the above transforms, e.g., for the traveling salesman problem. To the interest of this paper, we explain the techniques used in the second type of problems.

Given a universe $V$, let $\mathcal{R}$ be a ring and consider functions from $2^V$ to $\mathcal{R}$. Denote the collection of such functions by $\mathcal{R}[2^V]$. A singleton $f_A[X]$ is an element of $\mathcal{R}[2^V]$ which is zero unless $X=A$. The operator $\oplus$ is the pointwise addition and the operator $\odot$ is the pointwise multiplication. We first define some useful algebraic transforms.

The {\it zeta transform} of a function $f\in \mathcal{R}[2^V]$ is defined to be
\begin{equation}
\label{zeta}
\zeta f[Y] = \sum_{X\subseteq Y} f[X].
\end{equation}

The {\it M\"{o}bius transform/inversion} \cite{mobiusorigin,stanley2000enumerative} of $f$ is defined to be
\begin{equation}
\label{mobius}
\mu f[Y]=\sum_{X\subseteq Y} (-1)^{|Y\setminus X|}f[X].
\end{equation}

The M\"{o}bius transform is the inverse transform of the zeta transform, as they have the following relation \cite{mobiusorigin,stanley2000enumerative}:
\begin{equation}
\label{zetamobius}
\mu(\zeta f)[X]=f[X].
\end{equation}

The high level idea of \cite{savespace2010} is that, rather than directly computing $f[V]$ by storing exponentially many intermediate results $\{f[S]\}_{S\subseteq V}$, they compute the zeta transform of $f[S]$ using only polynomial space. $f[V]$ can be obtained by M\"{o}bius inversion (\ref{mobius}) as $f[V]=\sum_{X\subseteq V} (-1)^{|V\setminus X|}(\zeta f)[X]$. Problems which can be solved in this manner have a common nature. They have recurrences which can be formulated by subset convolutions. The {\it subset convolution} \cite{fouriermobius} is defined to be
\begin{equation}
\label{subsetconvolution}
f*_{\mathcal{R}} g[X]=\sum_{X'\subseteq X} f(X')g(X\setminus X').
\end{equation}

To apply the zeta transform to $f*_{\mathcal{R}}g$, we need the {\it union product} \cite{fouriermobius} which is defined as
\begin{equation}
\label{unionproduct}
f*_u g[X]=\sum_{X_1\bigcup X_2=X} f(X_1)g(X_2).
\end{equation}

The relation between the union product and the zeta transform is as follows \cite{fouriermobius}:
\begin{equation}
\label{unionproductzeta}
\zeta(f*_u g)[X]= (\zeta f)\odot(\zeta g)[X].
\end{equation}

In \cite{savespace2010}, functions over $(\mathcal{R}[2^V];\oplus,*_{\mathcal{R}})$ are modeled by arithmetic circuits. Such a circuit is a directed acyclic graph where every node is either a singleton (constant gate), a $\oplus$ gate or a $*_{\mathcal{R}}$ gate.
Given any circuit $C$ over $(\mathcal{R}[2^V];\oplus,*_{\mathcal{R}})$ which outputs $f$, every gate in $C$ computing an output $a$ from its inputs $b,c$ is replaced by small circuits computing a relaxation $\{a^i\}_{i=1}^{|V|}$ of $a$ from relaxations $\{b^i\}_{i=1}^{|V|}$ and $\{c^i\}_{i=1}^{|V|}$ of $b$ and $c$ respectively. (A {\it relaxation} of a function $f\in\mathcal{R}[2^V]$ is a sequence of functions $\{f^i:f^i\in\mathcal{R}[2^V], 0\leq i\leq |V| \}$, such that $\forall i, X\subseteq V$, $f^i[X]=f[X]$ if $i=|X|$, $f^i[X]=0$ if $i<|X|$, and $f^i[X]$ is an arbitrary value if $i>|X|$.)
For a $\oplus$ gate, replace $a=b\oplus c$ by $a^i=b^i\oplus c^i$, for $0\leq i\leq |V|$. For a $*_{\mathcal{R}}$ gate, replace $a=b *_{\mathcal{R}} c$ by $a^i=\sum_{j=0}^i b^j *_u c^{i-j}$, for $0\leq i\leq |V|$. This new circuit $C_1$ over $(\mathcal{R}[2^V];\oplus,*_u)$ is of size $O(|C|\cdot |V|)$ and outputs $f_{|V|}[V]$. The next step is to replace every $*_u$ gate by a gate $\odot$ and every constant gate $a$ by $\zeta a$. It turns $C_1$ to a circuit $C_2$ over $(\mathcal{R}[2^V]; \oplus, \odot)$, such that for every gate $a\in C_1$, the corresponding gate in $C_2$ outputs $\zeta a$. Since additions and multiplications in $C_2$ are pointwise, $C_2$ can be viewed as $2^{|V|}$ disjoint circuits $C^Y$ over $(\mathcal{R}[2^V]; +, \cdot)$ for every subset $Y\subseteq V$. The circuit $C^Y$ outputs $(\zeta f)[Y]$. It is easy to see that the construction of every $C^Y$ takes polynomial time.

As all problems of interest in this paper work on the integer domain $\mathbb{Z}$, we consider $\mathcal{R}=\mathbb{Z}$ and replace $*_{\mathcal{R}}$ by $*$ for simplicity. Assume $0 \leq f[V] < m$ for some integer $m$, we can view the computation as on the finite ring $\mathbb{Z}_{m}$. Additions and multiplications can be implemented efficiently on $\mathbb{Z}_{m}$ (e.g., using the fast algorithm in \cite{multiplication} for multiplication).

\begin{theorem}[Theorem 5.1 \cite{savespace2010}]
\label{thmsavespace}
Let $C$ be a circuit over $(\mathbb{Z}[2^V]; \oplus, *)$ which outputs $f$. Let all constants in $C$ be singletons and let $f[V] < m$ for some integer $m$. Then $f[V]$ can be computed in time $O^*(2^{|V|})$ and space $O(|V||C|\log m)$.
\end{theorem}

\subsection{Tree decomposition}

For any graph $G=(V,E)$, a {\it tree decomposition} of $G$ is a tree $\mathcal{T}=(V_{\mathcal{T}}, E_{\mathcal{T}})$ such that every node $x$ in $V_\mathcal{T}$ is associated with a set $B_x$ (called the bag of $x$) of vertices in $G$ and $\mathcal{T}$ has the following additional properties:

1. For any nodes $x, y$, and any node $z$ belonging to the path connecting $x$ and $y$ in $\mathcal{T}$, $B_x\cap B_y\subseteq B_z$.

2. For any edge $e=\{u, v\}\in E$, there exists a node $x$ such that $u, v\in B_x$. 

3. $\cup_{x\in V_{\mathcal{T}}} B_x = V$.

The {\it width} of a tree decomposition $\mathcal{T}$ is $\max_{x\in V_{\mathcal{T}}} |B_x|-1$. The {\it treewidth} of a graph $G$ is the minimum width over all tree decompositions of $G$. We reserve the letter $k$ for treewidth in the following context. Constructing a tree decomposition with minimum treewidth is an NP-hard problem. If the treewidth of a graph is bounded by a constant, a linear time algorithm for finding the minimum treewidth is known \cite{twconst}. An $O(\log n)$ approximation algorithm of the treewidth is given in \cite{twapproxlogn}. The result has been further improved to $O(\log k)$ in \cite{twapproxlogk}. There are also a series of works studying constant approximation of treewidth $k$ with running time exponential in $k$, see \cite{twconst} and references therein.

To simplify the presentation of dynamic programming based on tree decomposition, an arbitrary tree decomposition is usually transformed into a {\it nice} tree decomposition which has the following additional properties. A node in a nice tree decomposition has at most 2 children. Let $c$ be the only child of $x$ or let $c_1,c_2$ be the two children of $x$. Any node $x$ in a nice tree decomposition is of one of the following five types:
\begin{enumerate}
\item An {\it introduce vertex} node (introduce vertex $v$), where $B_x=B_c\cup\{v\}$.

\item An {\it introduce edge} node (introduce edge $e=\{u,v\}$), where $u,v\in B_x$ and $B_x=B_c$. We say that $e$ is associated with $x$.

\item A {\it forget vertex} node (forget vertex $v$), where $B_x=B_c\setminus \{v\}$.

\item A {\it join} node, where $x$ has two children and $B_x=B_{c_1}=B_{c_2}$.

\item A {\it leaf} node, a leaf of $\mathcal{T}$.
\end{enumerate}
For any tree decomposition, a nice tree decomposition with the same treewidth can be constructed in polynomial time \cite{nicetree}. Notice that an introduce edge node is not a type of nodes in a common definition of a nice tree decomposition. We can create an introduce edge node after the two endpoints are introduced. We further transform every leaf node and the root to a node with an empty bag by adding a series of introduce nodes or forget nodes respectively.

\section{Algorithmic framework}

We explain the algorithmic framework using the problem of counting perfect matchings based on tree decomposition as an example to help understand the recurrences. The result can be easily applied to other problems. A {\it perfect matching} in a graph $G=(V,E)$ is a collection of $|V|/2$ edges such that every vertex in $G$ belongs to exactly one of these edges.

Consider a connected graph $G$ and a nice tree decomposition $\mathcal{T}$ of treewidth $k$ on $G$. Consider a function $f\in \mathbb{Z}[2^V]$. Assume that the recurrence for computing $f$ on a join node can be formulated as a subset convolution, while on other types of tree nodes it is an addition or subtraction. We explain how to efficiently evaluate $f[V]$ on a nice tree decomposition by dynamic programming in polynomial space.
Let $\mathcal{T}_x$ be the subtree rooted at $x$. Let $T_x$ be the vertices contained in bags associated with nodes in $\mathcal{T}_x$ which are not in $B_x$. For any $X\subseteq B_x$, let $Y_X$ be the union of $X$ and $T_x$. For any $X\subseteq B_x$, let $f_x[X]$ be the number of perfect matchings in the subgraph $Y_X$ with edges introduced in $\mathcal{T}_x$.
As in the construction of Theorem \ref{thmsavespace}, we first replace $f_x$ by a relaxation $\{f_x^i\}_{0\leq i\leq k+1}$ of $f$, where $k$ is the treewidth. We then compute the zeta transform of $f_x^i$, for $0\leq i\leq k+1$. In the following context, we present only recurrences of $f_x$ for all types of tree nodes except the join node where we need to use the relaxations. The recurrences of $f_x$ based on $f_c$ can be directly applied to their relaxations with the same index as in Theorem \ref{thmsavespace}.

For any leaf node $x$, $(\zeta f_x)[\emptyset]=f_x[\emptyset]$ is a problem-dependent constant. In the case of the number of perfect matchings, $f_x[\emptyset]=1$. For the root $x$, $(\zeta f_x)[\emptyset] = f_x[\emptyset]= f[V]$ which is the value of interest. For the other cases, consider an arbitrary subset $X\subseteq B_x$.

1. $x$ is an introduce vertex node. If the introduced vertex $v$ is not in $X$, $f_x[X]=f_c[X]$. If $v\in X$, in the case of the number of perfect matchings, $v$ has no adjacent edges, hence $f_x[X]=0$ (for other problems, $f_x[X]$ may equal to $f_c[X]$, which implies a similar recurrence). By definition of the zeta transform, if $v\in X$, we have $(\zeta f_x)[X]=\sum_{v\in X'\subseteq X}f_x[X']+\sum_{v\notin X'\subseteq X}f_x[X']=\sum_{v\notin X'\subseteq X}f_x[X']$. Therefore,
\begin{eqnarray}
\label{introvertex}
(\zeta f_x)[X] =\left\{ \begin{array}{ll}
(\zeta f_c)[X] & \textrm{  $v\notin X$} \\
(\zeta f_c)[X\setminus\{v\}] & \textrm{  $v\in X$}
\end{array} \right.
\end{eqnarray}

2. $x$ is a forget vertex node. $f_x[X]=f_c[X\cup\{v\}]$ by definition.
\begin{eqnarray}
\label{forget}
    (\zeta f_x)[X]&=&\sum_{X'\subseteq X}f_x[X']=\sum_{X'\subseteq X}f_c[X'\cup\{v\}] \nonumber \\
    &=&(\zeta f_c)[X\cup\{v\}]-(\zeta f_c)[X].
\end{eqnarray}

3. $x$ is a join node with two children. By assumption, the computation of $f_x$ on a join node can be formulated as a subset convolution. We have
\begin{eqnarray}
\label{convolution}
   f_x[X]=\sum_{X'\subseteq X} f_{c_1}[X']f_{c_2}[X\setminus X']=f_{c_1}* f_{c_2}[X].
\end{eqnarray}
For the problem of counting perfect matchings, it is easy to verify that $f_x[X]$ can be computed using (\ref{convolution}). Let $f_x^i=\sum_{j=0}^i f_{c_1}^j*_u f_{c_2}^{i-j}$. We can transform the computation to
\begin{equation}
\label{join}
(\zeta f_x^i)[X]= \sum_{j=0}^i(\zeta f_{c_1}^j)[X]\cdot (\zeta f_{c_2}^{i-j})[X], \textrm{ for }0\leq i\leq k+1.
\end{equation}

4. $x$ is an introduce edge node introducing $e=\{u,v\}$. The recurrence of $f_x$ with respect to $f_c$ is problem-dependent. Since the goal of the analysis of this case is to explain why we need to modify the construction of an introduce edge node, we consider only the recurrence for the counting perfect matchings problem. In this problem, if $e\nsubseteq X$, $f_x[X]=f_c[X]$, then $(\zeta f_x)[X]  = (\zeta f_c)[X]$. If $e\subseteq X$, we can match $u$ and $v$ by $e$ or not use $e$ for matching, thus $f_x[X]=f_c[X]+f_c[X\setminus\{u,v\}]$. In this case, we have
\begin{eqnarray}
(\zeta f_x)[X] &=& \sum_{e\subseteq X'\subseteq X}f_x[X']+\sum_{e\nsubseteq X'\subseteq X}f_x[X'] = \sum_{e\subseteq X'\subseteq X}(f_c[X']+f_c[X'\setminus \{u,v\}]) \nonumber \\
&+&\sum_{e\nsubseteq X'\subseteq X}f_c[X'] = \sum_{X'\subseteq X}f_c[X']+\sum_{e\subseteq X'\subseteq X}f(X'\setminus\{u,v\}). \nonumber
\end{eqnarray}
Hence,
\begin{eqnarray}
\label{introedge}
(\zeta f_x)[X] =\left\{ \begin{array}{ll}
(\zeta f_c)[X] & \textrm{  $e\nsubseteq X$} \\
(\zeta f_c)[X]+(\zeta f_c)[X\setminus\{u,v\}] & \textrm{  $e\subseteq X$}
\end{array} \right.
\end{eqnarray}

In cases 2 and 4, we see that the value of $(\zeta f_x)[X]$ depends on the values of $\zeta f_c$ on two different subsets. We can visualize the computation along a path from a leaf to the root as a computation tree. This computation tree branches on introduce edge nodes and forget vertex nodes. Suppose along any path from the root to a leaf in $\mathcal{T}$, the maximum number of introduce edge nodes is $m'$ and the maximum number of forget vertex nodes is $h$. To avoid exponentially large storage for keeping partial results in this computation tree, we compute along every path from a leaf to the root in this tree. This leads to an increase of the running time by a factor of $O(2^{m'+h})$, but the computation is in polynomial space (explained in detail later). As $m'$ could be $\Omega(n)$, this could contribute a factor of $2^{\Omega(n)}$ to the time complexity. To reduce the running time, we eliminate the branching introduced by introduce edge nodes. On the other hand, the branching introduced by forget vertex nodes seems inevitable.

For any introduce edge node $x$ which introduces an edge $e$ and has a child $c$ in the original nice tree decomposition $\mathcal{T}$, we add an auxiliary child $c'$ of $x$, such that $B_{c'}=B_x$ and introduce the edge $e$ at $c'$. $c'$ is a special leaf which is not empty. We assume the evaluation of $\zeta f$ on $c'$ takes only polynomial time. For the counting perfect matchings problem, $f_{c'}[X]=1$ only when $X=e$ or $X=\emptyset$, otherwise it is equal to 0. Then $(\zeta f_{c'})[X]=2$ if $e\subseteq X$, otherwise $(\zeta f_{c'})[X]=1$. We will verify that this assumption is valid for other problems considered in the following sections. We call $x$ a {\it modified introduce edge} node and $c'$ an {\it auxiliary leaf}. As the computation on $x$ is the same as that on a join node, we do not talk about the computation on modified introduce edge nodes separately. \\

In cases 1 and 2, we observe that the addition operation is not a strictly pointwise addition as in Theorem \ref{thmsavespace}. This is because in a tree decomposition, the set of vertices on every tree node might not be the same. However, there is a one-to-one correspondence from a set $X$ in node $x$ to a set $X'$ in its child $c$. We call it a {\it relaxed pointwise addition} and denote it by $\oplus'$. Hence, $f$ can be evaluated by a circuit $C$ over $(\mathbb{Z}[2^V]; \oplus', *)$. We transform $C$ to a circuit $C_1$ over $(\mathbb{Z}[2^V]; \oplus', *_u)$, then to $C_2$ over $(\mathbb{Z}[2^V]; \oplus', \odot)$, following constructions in Theorem \ref{thmsavespace}.

In Theorem \ref{thmsavespace}, $C_2$ can be viewed as $2^{|V|}$ disjoint circuits. In the case of tree decomposition, the computation makes branches on a forget node. Therefore, we cannot take $C_2$ as $O(2^k)$ disjoint circuits. Consider a subtree $\mathcal{T}_x$ of $\mathcal{T}$ where the root $x$ is the only join node in the subtree. Take an arbitrary path from $x$ to a leaf $l$ and assume there are $h'$ forget nodes along this path. We compute along every path of the computation tree expanded by the path from $x$ to $l$, and sum up the result at the top. There are $2^{h'}$ computation paths which are independent. Hence we can view the computation as $2^{h'}$ disjoint circuits on $(\mathbb{Z}; +, \cdot)$. Assume the maximum number of forget nodes along any path from the root $x$ to a leaf in $\mathcal{T}_x$ is $h$ and there are $n_l$ leaves, the total computation takes at most $n_l\cdot 2^{h}$ time and in polynomial space.

In general, we proceed the computation in an in-order depth-first traversal on a tree decomposition $\mathcal{T}$. Every time we hit a join node $j$, we need to complete all computations in the subtree rooted at $j$ before going up. Suppose $j_{1},j_{2}$ are the closest join nodes in two subtrees rooted at the children of $j$ (if there is no other join node consider $j_1$ or $j_2$ to be empty). Assume there are at most $h_j$ forget nodes between $j,j_1$ and $j,j_2$. Let $T_x$ be the time to complete the computation of $(\zeta f_x)[X]$ at node $x$. We have $T_j\leq 2\cdot 2^{h_j}\cdot\max\{T_{j_1},T_{j_2}\})$. The modified edge node is a special type of join node. In this case, since one of its children $c_1$ is always a leaf, the running time only depends on the subtree rooted at $c_2$, thus similar to an introduce vertex node. Suppose there are $n_j$ join nodes and let $h$ be the maximum number of forget nodes along any path from the root to a leaf. By induction, it takes $2^{n_j}\cdot 2^h$ time to complete the computation on $\mathcal{T}$ and in polynomial space. Notice that $2^{n_j}$ is the number of leaves in $\mathcal{T}$, hence $2^{n_j}=O(|V|+|E|)$.

To summarize, we present the algorithm for the problem of counting perfect matchings based on a modified nice tree decomposition $\mathcal{T}$ in Algorithm 1.

\begin{algorithm}
\caption{Counting perfect matchings on a modified nice tree decomposition}
\begin{algorithmic}
\State {\bf Input}: a modified nice tree decomposition $\mathcal{T}$ with root $r$.
\State {\bf return} $(\zeta f)(r, \emptyset,0)$.
\Procedure {$(\zeta f)$}{$x$, $X$,$i$}. // $(\zeta f)(x, X,i)$ represents $(\zeta f_x^i)[X]$.
\State {\bf if} $x$ is a leaf: {\bf return} 1.
\State {\bf if} $x$ is an auxiliary leaf: {\bf return} 2 when $e\subseteq X$, otherwise 1.
\State {\bf if} $x$ is an introduce vertex node: {\bf return} $(\zeta f)(c, X,i)$ when $v\notin X$, or $(\zeta f)(c, X-\{v\},i)$ when $v\in X$.
\State {\bf if} $x$ is a forget vertex node: {\bf return} $(\zeta f)(c, X\cup\{v\},i)-(\zeta f)(c, X,i)$.
\State {\bf if} $x$ is a join node: {\bf return} $\sum_{j=0}^i(\zeta f)(c_1, X,j)\cdot (\zeta f)(c_2, X,i-j)$.
\EndProcedure
\end{algorithmic}
\end{algorithm}

For any tree decomposition $\mathcal{T}$ of a graph $G$, we can transform it to a modified nice tree decomposition $\mathcal{T}'$ with the convention that the root has an empty bag. In this way, the parameter $h$, the maximum number of forget nodes along any path from the root to a leaf in $\mathcal{T}'$ is equal to the maximum size of the union of all bags along any path from the root to a leaf in $\mathcal{T}$. We directly tie this number $h$ to the complexity of our algorithm. Let $h_m(G)$ be the minimum value of $h$ for all tree decompositions of $G$. We show that $h_m(G)$ is closely related to a well-known parameter, the {\it tree-depth} of a graph \cite{treedepth}.

\begin{definition}[tree-depth \cite{treedepth}]
Given a rooted tree $T$ with vertex set $V$, a closure of $T$, $clos(T)$ is a graph $G$ with the same vertex $V$, and for any two vertices $x,y\in V$ such that $x$ is an ancestor of $y$ in $T$, there is a corresponding edge $(x,y)$ in $G$. The tree-depth of $T$ is the height of $T$. The tree-depth of a graph $G$, $td(G)$ is the minimum height of trees $T$ such that $G\subseteq clos(T)$.
\end{definition}

\begin{Proposition}
For any connected graph $G$, $h_m(G) = td(G)$.
\end{Proposition}
\begin{proof}
For any tree decomposition of $G$, we first transform it to a modified nice tree decomposition $\mathcal{T}$. We contract $\mathcal{T}$ by deleting all nodes except the forget nodes. Let $T_f$ be this contracted tree such that for every forget node in $\mathcal{T}$ which forgets a vertex $x$ in $G$, the corresponding vertex in $T_f$ is $x$. We have $G\subseteq clos(T_f)$. Therefore, $td(G)\leq h$, here $h$ is the maximum number of forget nodes along any path from the root to a leaf in $\mathcal{T}$.

For any tree $T$ such that $G\subseteq clos(T)$, we construct a corresponding tree decomposition $\mathcal{T}$ of $G$ such that, $\mathcal{T}$ is initialized to be $T$ and every bag associated with the vertex $x$ of $T$ contains the vertex itself. For every vertex $x\in T$, we also put all ancestors of $x$ in $T$ into the bag associated with $x$. It is easy to verify that it is a valid tree decomposition of $G$. Therefore, the tree-depth of $T$, $td(T)\geq h_m(G)$. $\square$
\end{proof}

In the following context, we also call the parameter $h$, the maximum size of the union of all bags along any path from the root to a leaf in a tree decomposition $\mathcal{T}$, the tree-depth of $\mathcal{T}$.
Let $k$ be the treewidth of $G$, it is shown in \cite{treedepth} that $td(G)\leq (k+1)\log |V|$. Therefore, we also have $h_m(G)\leq (k+1)\log |V|$. Moreover, it is obvious to have $h_m(G)\geq k+1$.

Finally, we summarize the main result of this section in the following theorem.
\begin{theorem}
Given any graph $G=(V, E)$ and tree decomposition $\mathcal{T}$ on $G$. Let $f$ be a function evaluated by a circuit $C$ over $(\mathbb{Z}[2^V]; \oplus', \ast)$ with constants being singletons. Assume $f[V]< m$ for integer $m$. We can compute $f[V]$ in time $O^*((|V|+|E|)2^{h})$ and in space $O(|V||C|\log m)$. Here $h$ is the maximum size of the union of all bags along any path from the root to a leaf in $\mathcal{T}$.
\end{theorem}


\section{Counting perfect matchings}

The problem of counting perfect matchings is $\sharp$P-complete. It has long been known that in a bipartite graph of size $2n$, counting perfect matchings takes $O^*(2^n)$ time using the inclusion and exclusion principle. A recent breakthrough \cite{pmasryser} shows that the same running time is achievable for general graphs. For low degree graphs, improved results based on dynamic programming on path decomposition on a sufficiently large subgraph are known \cite{matchingpwsubgraph}.

Counting perfect matchings on grids is an interesting problem in statistical physics \cite{monomer}. The more generalized problem is the Monomer-Dimer problem \cite{monomer}, which essentially asks to compute the number of matchings of a specific size. We model the Monomer-Dimer problem as computing the matching polynomial problem . For grids in dimension 2, the pure Dimer (perfect matching) problem is polynomial-time tractable and an explicit expression of the solution is known \cite{matching2dim}. We consider the problem of counting perfect matchings in cube/hypercube in Section 4.1. Results on counting perfect matchings in more general grids, computing the matching polynomial and applications to other set covering and partitioning problems are presented in Section 4.2.

\subsection{Counting perfect matchings on cube/hypercube}

We consider the case of counting perfect matchings on grids of dimension $d$, where $d\geq 3$ and the length of the grid is $n$ in each dimension. We denote this grid by $G_d(n)$. To apply Algorithm 1, we first construct a balanced tree decomposition on $G_d(n)$ with the help of balanced separators. The balanced tree decomposition can easily be transformed into a modified nice tree decomposition.

\textbf{Tree decomposition using balanced vertex separators.} We first explain how to construct a balanced tree decomposition using vertex separators of general graphs. An $\alpha$-balanced vertex separator of a graph/subgraph $G$ is a set of vertices $S\subseteq G$, such that after removing $S$, $G$ is separated into two disjoint parts $A$ and $B$ with no edge between $A$ and $B$, and $|A|, |B|\leq \alpha|G|$, where $\alpha$ is a constant in $(0,1)$.
Suppose we have an oracle to find an $\alpha$-balanced vertex separator of a graph. We begin with creating the root of a tree decomposition $\mathcal{T}$ and associate the vertex separator $S$ of the whole graph with the root. Consider a subtree $\mathcal{T}_x$ in $\mathcal{T}$ with the root $x$ associated with a bag $B_x$. Denote the vertices belonging to nodes in $\mathcal{T}_x$ by $V_x$. Initially, $V_x=V$ and $x$ is the root of $\mathcal{T}$. Suppose we have a vertex separator $S_x$ which partitions $V_x$ into two disjoint parts $V_{c_1}$ and $V_{c_2}$. We create two children $c_1,c_2$ of $x$, such that the set of vertices belonging to $\mathcal{T}_{c_i}$ is $S_x\cup V_{c_i}$. Denote the set of vertices belonging to nodes in the path from $x$ to the root of $\mathcal{T}$ by $U_x$, we define the bag $B_{c_i}$ to be $S_x\cup (V_{c_i}\cap U_x)$, for $i=1,2$. It is easy to verify that this is a valid tree decomposition. Since $V_x$ decreases by a factor of at least $1-\alpha$ in each partition, the height of the tree is at most $\log_{\frac{1}{1-\alpha}} n$. To transform this decomposition into a modified nice tree decomposition, we only need to add a series of introduce vertex nodes, forget vertex nodes or modified introduce edge nodes between two originally adjacent nodes. We call this tree decomposition algorithm {\bf Algorithm 2}. \\

We observe that after the transformation, the number of forget nodes from $B_{c_i}$ to $B_x$ is the size of the balanced vertex separator of $V_x$, i.e. $|S_x|$. Therefore, the number of forget nodes from the root to a leaf is the sum of the sizes of the balanced vertex separators used to construct this path in the tree decomposition.

A grid graph $G_d(n)$ has a nice symmetric structure. Denote the $d$ dimensions by $x_1,x_2,...,x_d$ and consider an arbitrary subgrid $G'_d$ of $G_d(n)$ with length $n'_i$ in dimension $x_i$. The hyperplane in $G_d'$ which is perpendicular to $x_i$ and cuts $G'_d$ into halves can be used as a $1/2$-balanced vertex separator. We always cut the dimension with the longest length. If $n_i'=n_{i+1}'$, we choose to first cut the dimension $x_i$, then $x_{i+1}$. We illustrate the construction of the 2-dimensional case in the following example.

\begin{example}[Balanced tree decomposition on $G_2(n)$]
\label{exp2dgrid}
The left picture is a partitioning on a 2-dimensional grid. We always bipartition the longer side of the grid/subgrid. The right picture is the corresponding balanced tree decomposition on this grid. The same letters on both sides represent the same set of nodes. $P_i$ represent a balanced vertex separator. We denote the left/top half of $P_i$ by $P_{i1}$, and the right/bottom part by $P_{i2}$ (see Figure \ref{example2dgrid}). The treewidth of this decomposition is $\frac{3}{2}n$.
\begin{figure}[!t]
\centering
\includegraphics[width=5in]{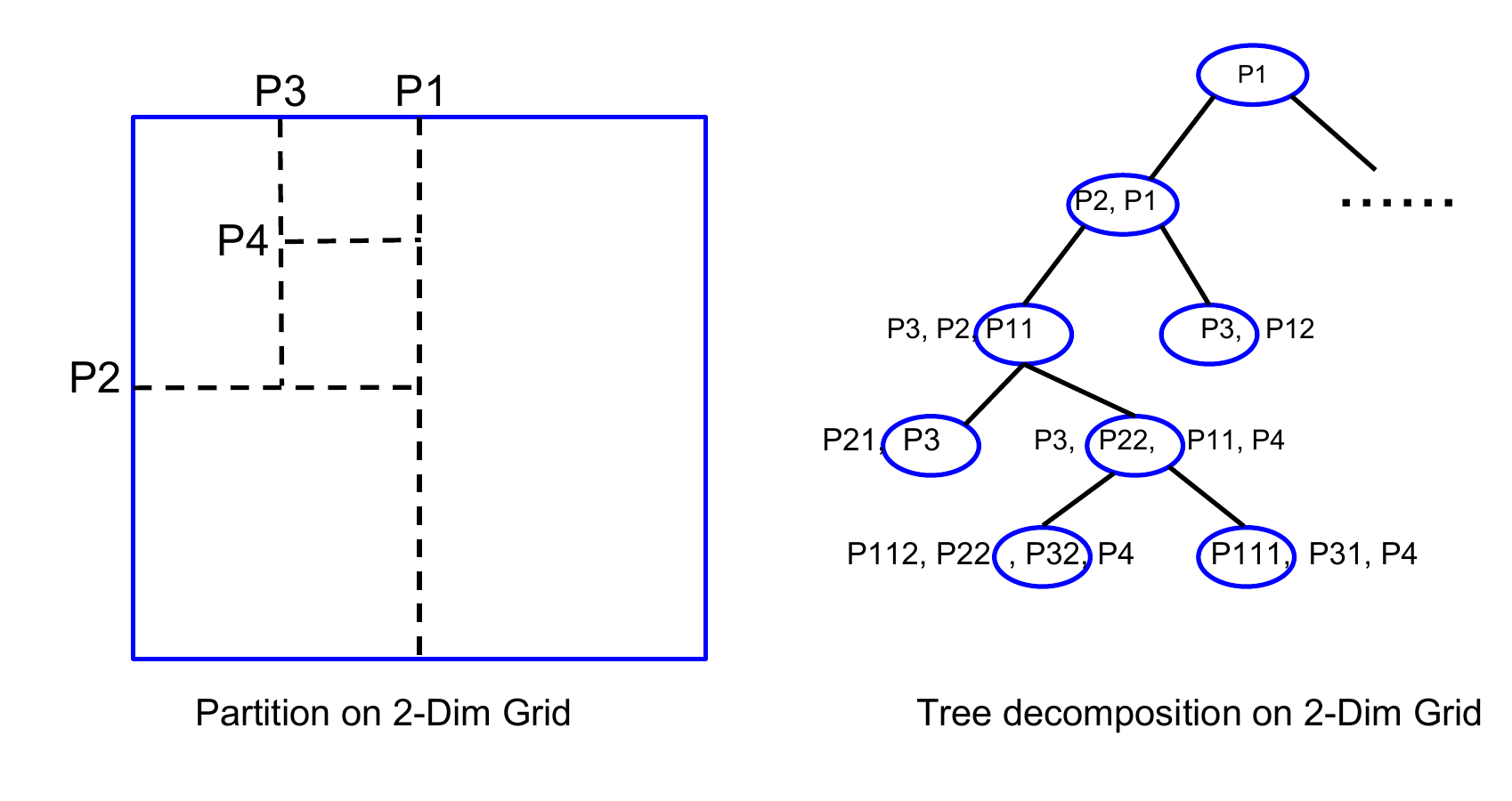}
\caption{An illustrative figure for balanced tree decomposition on $G_2(n)$.}
\label{example2dgrid}
\end{figure}
\end{example}

To run Algorithm 2 on $G_d(n)$, we cut dimensions $x_1,x_2,...,x_d$ consecutively with separators of size $\frac{1}{2^{i-1}}n^{d-1}$, for $i=1,2...,d$. Then we proceed with subgrids of length $n/2$ in every dimension. It is easy to see that the treewidth of this tree decomposition is $\frac{3}{2}n^{d-1}$. The tree-depth $h$ of this tree decomposition is at most $\sum_{j=0}^{\infty}\sum_{i=0}^{d-1} \frac{1}{2^i}\cdot (\frac{1}{2^j}n)^{d-1}$, which is $\frac{2^d-1}{2^{d-1}-1}n^{d-1}$.

\begin{lemma}
The treewidth of the tree decomposition $\mathcal{T}$ on $G_d(n)$ obtained by Algorithm 2 is $\frac{3}{2}n^{d-1}$. The tree-depth of $\mathcal{T}$ is at most $\frac{2^d-1}{2^{d-1}-1}n^{d-1}$.
\end{lemma}

To apply Algorithm 1 to the problem of counting perfect matchings, we verify that $f[S]\leq {|E|\choose |V|/2}\leq |E|^{|V|/2}$ and all constants are singletons.

\begin{theorem}
The problem of counting perfect matchings on grids of dimension $d$ and uniform length $n$ can be solved in time $O^*(2^{\frac{2^d-1}{2^{d-1}-1}n^{d-1}})$ and in polynomial space.
\end{theorem}

To the best of our knowledge, there is no rigorous time complexity analysis of the counting perfect matchings problem in grids in the literature. To demonstrate the efficiency of Algorithm 1, we compare it to three other natural algorithms.

\textbf{1. Dynamic programming based on path decomposition.} A path decomposition is a special tree decomposition where the underlying tree is a path. A path decomposition with width $2n^{d-1}$ is obtained by putting all vertices with $x_1$ coordinate equal to $j$ and $j+1$ into the bag of node $j$, for $j=0,1,...,n-1$. A path decomposition with a smaller pathwidth of $n^{d-1}$ can be obtained as follows. Construct $n$ nodes $\{p_1,p_2,...,p_n\}$ associated with a bag of vertices with $x_1$ coordinate equal to $j$, for $j=0,1,...,n-1$. For any $p_j,p_{j+1}$, start from $p_j$, add a sequence of nodes by alternating between adding a vertex of $x_1=j+1$ and deleting its neighbor with $x_1=j$. The number of nodes increases by a factor of $n^{d-1}$ than the first path decomposition. We run the standard dynamic programming on the second path decomposition. This algorithm runs in time $O^*(2^{n^{d-1}})$, however the space complexity is $O^*(2^{n^{d-1}})$. It is of no surprise that it has a better running time than Algorithm 1 due to an extra space usage. We remark that van Rooij et al. \cite{dpgeneralsubsetconvolution} give a dynamic programming algorithm for the counting perfect matching problem on any tree decomposition of treewidth $k$ with running time $O^*(2^k)$ and space exponential to $k$. \\


\textbf{2. Dynamic programming based on path decomposition on a subgrid.} One way to obtain a polynomial space dynamic programming is to construct a low pathwidth decomposition on a sufficiently large subgraph. One can then run dynamic programming on this path decomposition and do an exhaustive enumeration on the remaining graph in a similar way as in \cite{matchingpwsubgraph}. To extract from $G_d(n)$ a subgrid of pathwidth $O(\log n)$ (notice that this is the maximum pathwidth for a polynomial space dynamic programming algorithm), we can delete a portion of vertices from $G_d(n)$ to turn a "cube"-shaped grid into a long "stripe" with $O(\log n)$ cross-section area. It is sufficient to remove $O(\frac{n^d}{(\log n)^{1/(d-1)}})$ vertices. This leads to a polynomial-space algorithm with running time $2^{O(\frac{n^d}{(\log n)^{1/(d-1)}})}$, which is worse than Algorithm 1. \\

\textbf{3. Branching algorithm.} A naive branching algorithm starting from any vertex in the grid could have time complexity $2^{O(n^d)}$ in the worst case. We analyze a branching algorithm with a careful selection of the starting point. The branching algorithm works by first finding a balanced separator $S$ and partitioning the graph into $A\cup S\cup B$. The algorithm enumerates every subset $X\subseteq S$. A vertex in $X$ either matches to vertices in $A$ or to vertices in $B$ while vertices in $S\setminus X$ are matched within $S$. Then the algorithm recurses on $A$ and $B$. Let $T_d(n)$ be the running time of this branching algorithm on $G_d(n)$. We use the same balanced separator as in Algorithm 2. We have an upper bound of the running time as, $T_d(n)\leq 2T_{d}(\frac{n-|S|}{2})\sum_{X\subseteq S} 2^{|X|} T_{d-1}(|S\setminus X|)$. We can use any polynomial space algorithm to count perfect matchings on $S\setminus X$. For example using Algorithm 1, since the separator is of size $O(n^{d-1})$, we have $T_{d-1}(|S\setminus X|)=2^{O(n^{d-2})}$. Therefore, $T_d(n)\leq 2T_{d}(\frac{n}{2})\cdot 2^{o(n^{d-1})}\sum_{i=0}^{|S|}{|S|\choose i} 2^{i}=2T_{d}(\frac{n}{2})\cdot 2^{o(n^{d-1})} 3^{|S|}$. We get $T_d(n) = O^*(3^h)$, i.e. $O^*(3^{\frac{2^d-1}{2^{d-1}-1}n^{d-1}})$, which is worse than Algorithm 1. We remark that this branching algorithm can be viewed as a divide and conquer algorithm on balanced tree decomposition, which is similar as in \cite{kpath}.

\subsection{Extensions}

\textbf{Counting perfect matchings on general grids. }
Consider more general grids of dimension $d$ with each dimension of length $n_i$, $1\leq i\leq d$, which is at most $n_m$. We use Algorithm 2 to construct a balanced tree decomposition $\mathcal{T}$ of a general grid and obtain an upper bound of the tree-depth $h$ of $\mathcal{T}$. 
\begin{lemma}
\label{lemmagrid}
Given any grid of dimension $d$ and volume $\mathcal{V}$. Using Algorithm 2, the tree-depth of this tree decomposition is at most $\frac{3d\mathcal{V}}{n_m}$.
\end{lemma}

\begin{proof}
Assume that $2^{q_i-1}n_{i+1}< n_i\leq 2^{q_i}n_{i+1}$ for some integer $q_i\geq 0$ and $i=1,2,...,d-1$. Let $h(q_1,...,q_{d-1})$ be the maximum number of forget nodes from the root to a leaf in this case. We can think of the whole construction as in $d$ phases (the algorithm might do nothing in some phases).

In Phase 1, the grid/subgrid is halved in dimension $x_1$ in $q_i$ times. For $i=2,...,d$, suppose the lengths of dimension $x_1,x_2,...,x_{i-1},x_{i}$ are $n_1',n_2',...,n_{i-1}',n_i'=n_i$ respectively, we have $n_i'/2<n_1'\leq n_2'\leq\cdots\leq n_{i-1}'\leq n_i'$. For any $1\leq j\leq i-1$, if $n_j'=n_{j+1}'$, Algorithm 2 will first cut dimension $x_j$ then $x_{j+1}$. If $n_j' < n_{j+1}'$, Algorithm 2 will first cut dimension $x_{j+1}$ then $x_{j}$. In this way, we obtain a new partition order $x_1',...,x_i'$ which is a permutation of $x_1,...,x_i$. In Phase $i$ for $i\leq d-1$, the grid/subgrid is halved in dimension $x_1',x_2',...,x_i'$ consecutively in $q_i$ rounds. In Phase $d$, the algorithm repeats bipartitioning dimension $x_1',x_2',...,x_d'$ until the construction is completed. We denote the maximum number of forget nodes from the root to a leaf created in Phase $i$ by $h_i$. 

$i=1$. $h_1=\frac{q_1\mathcal{V}}{n_1}$. Notice that $n_1=2^{q_1+\cdots q_{d-1}}(\frac{\mathcal{V}}{2^{q_1+2q_2+\cdots +(d-1)q_{d-1}}})^{1/d}$, $h_1$ is maximized when $q_1=\frac{1}{\ln 2}\cdot\frac{d}{d-1}$ and $q_2=...=q_{d-1}=0$. We have $h_1\leq \frac{3\mathcal{V}}{n_1}$.

$i=2$. If $n_1=2^{q_1}n_{2}$, $h_2=\frac{\mathcal{V}}{n_1}\cdot(1+\frac{1}{2}+\cdots +\frac{1}{2^{q_2-1}})+\frac{\mathcal{V}}{n_2}\cdot(\frac{1}{2^{q_1+1}}+\frac{1}{2^{q_1+2}}+\cdots +\frac{1}{2^{q_1+q_2}})$, i.e. $h_2=\frac{\mathcal{V}}{n_1}\cdot\frac{2^2-1}{2-1}\cdot(1-\frac{1}{2^{(2-1)q_2}})\leq \frac{3\mathcal{V}}{n_1}$.

If $n_1< 2^{q_1}n_{2}$, Algorithm 2 will alternate to cut the $x_2$ dimension and $x_1$ dimension in $q_2$ rounds. $h_2=\frac{\mathcal{V}}{n_1}\cdot(\frac{1}{2}+\cdots +\frac{1}{2^{q_2}})+\frac{\mathcal{V}}{n_2}\cdot(\frac{1}{2^{q_1}}+\frac{1}{2^{q_1+1}}+\cdots +\frac{1}{2^{q_1+q_2-1}})$. Since $2^{q_1}n_2>n_1$, $h_2<\frac{\mathcal{V}}{n_1}\cdot(\frac{\frac{1}{2}(1-\frac{1}{2^{q_2}})}{1-1/2}+\frac{1-\frac{1}{2^{q_2}}}{1-1/2})< \frac{3\mathcal{V}}{n_1}$.

In general, for any $2\leq i\leq d-1$, we can bound $h_i$ as $h_i\leq \frac{\mathcal{V}}{2^{q_2+2q_3+\cdots+(i-2)q_{i-1}}n_1}\cdot(1+\frac{1}{2}+\cdots+\frac{1}{2^{i-1}})\cdot(1+\frac{1}{2^{i-1}}+\cdots+\frac{1}{2^{(i-1)(q_i-1)}})$. Hence, $h_i\leq \frac{\mathcal{V}}{2^{q_2+2q_3+\cdots+(i-2)q_{i-1}}n_1}\cdot\frac{2^i-1}{2^{i-1}-1}\cdot(1-\frac{1}{2^{(i-1)q_i}})$, which is at most $\frac{3\mathcal{V}}{n_1}$.


$i=d$. $h_d\leq \frac{2^{d}-1}{2^{d-1}-1}(\frac{\mathcal{V}}{2^{q_1+2q_2+\cdots +(d-1)q_{d-1}}})^{1-1/d}\leq \frac{3\mathcal{V}}{n_1}$.

Hence, $h(q_1,...,q_{d-1})=h_1+h_2+\cdots+h_d\leq \frac{3d\mathcal{V}}{n_1}$. $\square$
\end{proof}

Based on Lemma \ref{lemmagrid}, we give time complexity results of algorithms discussed in Section 4.1. First, $h$ is the only parameter to the running time of Algorithm 1 and the branching algorithm. Algorithm 1 runs in time $O^*(2^{\frac{3d\mathcal{V}}{n_m}})$ and the branching algorithm runs in time $O^*(3^{\frac{3d\mathcal{V}}{n_m}})$. The dynamic programming algorithm based on path decomposition on a subgrid has a running time $2^{O(\frac{\mathcal{V}}{(\log n_m)^{1/(d-1)}})}$. Those three algorithms have polynomial space complexity. For constant $d$, Algorithm~1 has the best time complexity. For the dynamic programming algorithm based on path decomposition, it runs in time $O^*(2^{\frac{\mathcal{V}}{n_m}})$ but in exponential space.

The result can easily be generalized to the $k$-nearest-neighbor ($k$NN) graphs and their subgraphs in $d$-dimensional space, as it is known that there exists a vertex separator of size $O(k^{1/d}n^{1-1/d})$ which splits the $k$NN graph into two disjoint parts with size at most $\frac{d+1}{d+2}n$ \cite{knn}.
More generally, we know that a nontrivial result can be obtained by Algorithm 1
if there exists a balanced separator of the graph $G$ with the following property. Let $s(n')$ be the size of a balanced separator $S$ on any subgraph $G'$ of $G$ of size $n'\leq n$. $S$ partitions the subgraph into two disjoint parts $G_1',G_2'$, such that $S\cup G_i'$ is of size at most $cn'$, for some constant $c\in(0,1)$, $i=1,2$. If there exists a constant $\gamma<1$, such that for every $n'\leq n$, $s(cn')\leq \gamma s(n')$, then the number of forget nodes along any path from the root to a leaf is at most $s(n)+\gamma s(n)+\gamma^2 s(n)+\cdots\leq \frac{s(n)}{1-\gamma}$.
In this case, the tree decomposition of treewidth $k$ constructed by Algorithm 2 has the tree-depth $h=\Theta(k)$. For $k=\Omega(\log n)$, Algorithm 1 has a similar running time as the standard dynamic programming algorithm but with much better space complexity. \\


\textbf{Computing the matching polynomial.} The matching polynomial of a graph $G$ is defined to be $m[G, \lambda]=\sum_{i=0}^{|G|/2} m^i[G]\lambda^i$, where $m^i[G]$ is the number of matchings of size $i$ in graph $G$. We put the coefficients of $m[G, \lambda]$ into a vector $\mathbf{m}[G]$. The problem is essentially to compute the coefficient vector $\mathbf{m}[G]$.

For every node $x$ in a tree decomposition, let vector $\mathbf{m}_x[X]$ be the coefficient vector of the matching polynomial defined on $Y_X$. Notice that every entry of $\mathbf{m}_x[X]$ is at most $|E|^{|V|/2}$ and all constants are singletons.
$\mathbf{m}^0_x[X]=1$ and $\mathbf{m}^i_x[X]=0$ for $i>|X|/2$. The case of $x$ being a forget vertex node follows exactly from Algorithm 1. For any type of tree node $x$, 

- $x$ is a leaf node. $\mathbf{m}^i_x[\emptyset]=1$ if $i=0$, or 0 otherwise.

- $x$ is an introduce vertex node. If $v\in X$, $\mathbf{m}^i_x[X] =\mathbf{m}^i_c[X\setminus \{v\}]$. Hence $(\zeta\mathbf{m}^i_x)[X]=2(\zeta\mathbf{m}^i_c)[X\setminus\{v\}]$ if $v\in X$, or $(\zeta\mathbf{m}^i_x)[X]=(\zeta\mathbf{m}^i_c)[X]$ otherwise.

- $x$ is an auxiliary leaf of a modified introduce edge node. $\mathbf{m}_{x}^i[X]=1$ only when $u,v\in X$ and $i=1$, or $i=0$. Otherwise it is 0.

- $x$ is a join node. $\mathbf{m}_x^i[X]=\sum_{X'\subseteq X}\sum_{j=0}^i \mathbf{m}_{c_1}^j[X']\mathbf{m}_{c_2}^{i-j}[X\setminus X']$. \\

{\bf Counting $l$-packings.} Given a universe $U$ of elements and a collection of subsets $\mathcal{S}$ on $U$, an $l$-packing is a collection of $l$ disjoint sets. The $l$-packings problem can be solved in a similar way as computing the matching polynomial. Packing problems can be viewed as matching problems on hypergraphs. Tree decomposition on graphs can be generalized to tree decomposition on hypergraph, where we require every hyperedge to be assigned to a specific bag \cite{hypertree}. A hyperedge is introduced after all vertices covered by this edge are introduced. \\

\textbf{Counting dominating sets, counting set covers.} The set cover problem is given a universe $U$ of elements and a collection of sets $\mathcal{S}$ on $U$, find a subcollection of sets from $\mathcal{S}$ which covers the entire universe $U$. The dominating set problem is defined on a graph $G=(V,E)$. Let $U=V$, $\mathcal{S}=\{N[v]\}_{v\in V}$, where $N[v]$ is the union of the neighbors of $v$ and $v$ itself. The dominating set problem is to find a subset of vertices $S$ from $V$ such that $\bigcup_{v\in S} N[v]$ covers $V$.

The set cover problem can be viewed as a covering problem on a hypergraph, where one selects a collection of hyperedges which cover all vertices. The dominating set problem is then a special case of the set cover problem. If $\mathcal{S}$ is closed under subsets, a set cover can be viewed as a disjoint cover. We only consider the counting set covers problem. For any subset $X\subseteq B_x$, we define $h_x[X]$ to be the number of set covers of $Y_X$. We have $h_x[X]\leq |U|^{|\mathcal{S}|}$, and all constants are singletons. We omit the recurrence for forget vertex nodes as we can directly apply recurrence (\ref{forget}) in Algorithm 1. For any node $x$, $h_x[\emptyset]=1$.

- $x$ is a leaf node. $h_x[\emptyset]=1$.

- $x$ is an introduce vertex node. If $v\in X$, $h_x[X]=0$. If $v\notin X$, $h_x[X]=h_c[X]$.

- $x$ is an auxiliary leaf of a modified introduce hyperedge node. $h_{x}[X]=1$ when $X\subseteq e$, and $h_x[X]=0$ otherwise.

- $x$ is a join node. $h_x[X]=\sum_{X'\subseteq X}h_{c_1}[X']h_{c_2}[X-X']$. \\



Finally, we point out that our framework has its limitations. First, it cannot be applied to problems where the computation on a join node cannot be formalized as a convolution. The maximum independent set problem is an example. Also it is not known if there is a way to adopt the framework to the Hamiltonian path problem, the counting $l$-path problems, and the unweighted Steiner tree problem. It seems that for theses problems we need a large storage space to record intermediate results. It is interesting to find more problems which fit in our framework.

\section{Conclusion}
\label{conclusion}

We study the problem of designing efficient dynamic programming algorithms based on tree decompositions in polynomial space. We show how to construct a modified nice tree decomposition $\mathcal{T}$ and extend the algebraic techniques in \cite{savespace2010} to dynamic sets such that we can run the dynamic programming algorithm in time $O^*(2^h)$ and in polynomial space, with $h$ being the maximum size of the union of bags along any path from the root to a leaf of $\mathcal{T}$, a parameter closely related to the tree-depth of a graph \cite{treedepth}. We apply our algorithm to many problems. It is interesting to find more natural graphs with nontrivial modified nice tree decompositions, and to find more problems which fit in our framework.

\bibliographystyle{plain}
\bibliography{twdp}

\end{document}